\documentclass[11pt]{llncs}
\usepackage{algorithm2e}
\usepackage{algorithmicx}
\usepackage{epsfig}
\usepackage{fancybox}
\usepackage{graphics}
\usepackage{geometry}

\usepackage{amssymb}

\begin{document}
\title{Multidimensional $\beta$-skeletons in $L_1$ and $L_{\infty}$ metric
\thanks{This research is supported by the ESF EUROCORES programme EUROGIGA,
CRP VORONOI.}}
\author{
        Miros{\l}aw Kowaluk  
            \and
        Gabriela Majewska  
\institute{
Institute of Informatics, University of Warsaw, Warsaw, Poland,\\
\texttt{kowaluk@mimum.edu.pl},
\texttt{gm248309@students.mimuw.edu.pl}
}}
\date{\today}

\maketitle

\begin{abstract}
The $\beta$-skeleton $\{G_{\beta}(V)\}$ for a point set V is a family
of geometric graphs, defined by the notion of neighborhoods parameterized 
by real number $0 < \beta < \infty$. 
By using the distance-based version definition of $\beta$-skeletons we study 
those graphs for a set of points in $\mathbb{R}^d$ space with $l_1$ and $l_{\infty}$ 
metrics. We present algorithms for the entire spectrum of $\beta$ values 
and we discuss properties of lens-based and circle-based $\beta$-skeletons 
in those metrics.
Let $V \in \mathbb{R}^d$ in $L_{\infty}$ metric be a set of $n$ points in general 
position. Then, for $\beta<2$ lens-based $\beta$-skeleton $G_{\beta}(V)$ can 
be computed in $O(n^2 \log^d n)$ time. For $\beta \geq 2$ there exists 
an $O(n \log^{d-1} n)$ time algorithm that constructs $\beta$-skeleton for the set $V$.
We show that in $\mathbb{R}^d$ with $L_{\infty}$ metric, for $\beta<2$ 
$\beta$-skeleton $G_{\beta}(V)$ for $n$ points can be computed in $O(n^2 \log^d n)$ 
time. For $\beta \geq 2$ there exists an $O(n \log^{d-1} n)$ time algorithm.
In $\mathbb{R}^d$ with $L_1$ metric for a set of $n$ points in arbitrary 
position $\beta$-skeleton $G_{\beta}(V)$ can be computed in $O(n^2 \log^{d+2} n)$ time.

\end{abstract}

\section{Introduction}
 
$\beta$-skeletons \cite{kr85} in $R^2$ belong to the family of
proximity graphs, geometric graphs in which two vertices (points) define an edge 
if and only if they satisfy particular geometric requirements. 
 
$\beta$-skeletons are both important and popular because of many practical applications 
which span a spectrum of areas including geographic information systems, wireless ad hoc 
networks and machine learning. 
They allow to reconstruct a shape of a two-dimensional object from a given set 
of sample points and they are also helpful in finding the minimum weight triangulation 
of a point set. Two different forms of $\beta$-neighborhoods have been studied 
for $\beta > 1$ (see  \cite{abe98,e02}) leading to two different families 
of $\beta$-skeletons: lens-based and circle-based ones. 
  
Well-known examples of lens-based $\beta$-skeletons include {\em Gabriel Graph} 
($1$-beta skeleton), defined by Gabriel and Sokal \cite{gs69} 
and {\em Relative Neighborhood Graph} $RNG$ (for $\beta=2$) which was introduced 
by Toussaint \cite{tou80}.

Many algorithms computing Gabriel graphs and relative neighborhood graphs 
in subquadratic time were created \cite{ms84,su83,jk87,jky89,l94,jt92}. 
Some of them can be also used in metrics different than euclidean ($L_p$ metrics 
for $1<p<\infty$).
In $L_1$ and $L_{\infty}$ metrics the problem of computing $2$-dimensional $RNG$ 
has been adressed by Lee \cite{l85} and O'Rourke \cite{o82}. Wulff-Nilsen \cite{w06} 
presented an algorithm that computes Gabriel graph in the plane with a fixed 
orientation metric. Maignan and Gruau \cite{gm11} defined metric Gabriel graphs 
which can be used in arbitrary metric space.\\
The concept of constructing $\beta$-skeletons in higher-dimensional space is still 
not well explored. Toussaint \cite{tou80} presented an algorithm that computes 
the relative neighborhood graph for a set of $n$ points in $\mathbb{R}^d$ in $O(n^3)$ 
time. Supowit \cite{su83} designed the first subcubic algorithm that computes the $RNG$ 
in $O(n^2)$ but only in case where no three points in the input set form an isosceles 
triangle. His method is based on partitioning the space around each point 
from the given set so that a linear size supergraph of RNG can be constructed. 
From a straightforward edge elimination of this supergraph a relative neighborhood 
graph can be computed. It has been shown that this algorithm extends to $L_p$ metrics 
(see \cite{gbt84,y82}). Based on those ideas Smith \cite{s89} provided an algorithm 
that computes the RNG in $\mathbb{R}^3$ for a set of $n$ points in general position 
in $O(n^{\frac{23}{12}}\log n)$ time. Jaromczyk and Kowaluk \cite{jk87} improved those 
results to obtain an $O(n^2)$ algorithm that works in $L_p, 1<p<\infty$ metrics 
for a set of points in general position. They also provided an algorithm \cite{jk91} 
that computes the relative neighborhood graph for an arbitrary set of $n$ points 
in $\mathbb{R}^3$ in $O(n^2 \log n)$.\\
An algorithm  presented by Katajainen and Nevalainen \cite{kn87} can be used 
in $d$-dimentional spaces in $L_p$ metrics for $1 \leq p \leq \infty$. 
In $3$-dimensional euclidean space it works in $O(n^{\frac{5}{2+\epsilon}})$ time. 

Agarwal and Matou\v{s}ek \cite{am92} used data structures based on the partitioning 
scheme of Chazelle et al. \cite{csw90} and an arrangement of spheres (see \cite{cegsw90}) 
to prove that for a set of $n$ points in general position in $\mathbb{R}^3$ 
the relative neighborhood graph can be constructed in $O(n^{\frac{3}{2}}+\epsilon)$, 
for every $\epsilon > 0$. They also provided an algorithm that computes $RNG$ 
of an arbitrary set of points in $\mathbb{R}^3$ in $O(n^{\frac{7}{4+\epsilon}})$.

O'Rourke \cite{o82} designed an algorithm computing the relative neighborhood graph 
for a set of $n$ points in $L_{\infty}$ in general position in $O(n^2 \log n)$ time. 
Smith \cite{s89} shed that in any $\mathbb{R}^d$ space with $L_{\infty}$ the $RNG$ 
can be constructed in $O(n(\log n)^{d-1})$ time, assuming that the points are 
in general position.

This paper is organized as follows.
The basic properties and definitions for $\beta$-skeletons in $\mathbb{R}^d$ space 
with $L_1$ and $L_{\infty}$ metric are presented in section 2. In section 3 we describe 
algorithms computing $\beta$-skeletons in $L_{\infty}$. In the next section we give 
an algorithm for constructing those graphs in $L_1$ metric.
The last section contains open problems and conclusions.

\section{Preliminaries}
Let $V$ be a set of $n$ points in $\mathbb{R}^d$ space with $L_1$ (or respectively 
$L_{\infty}$) metric.
For any two points $p_1,p_2 \in \mathbb{R}^d$ we define the set $S(p_1,p_2)$ which 
is the sum of all shortest paths between $p_1$ and $p_2$, where point 
$p$  belongs to $S(p_1,p_2)$ if and only if $d_1(p_1,p)+d_1(p,p_2)=d_1(p_1,p_2)$ 
($d_{\infty}(p_1,p)+d_{\infty}(p,p_2)=d_{\infty}(p_1,p_2)$ respectively).

By modifying the definition of lens-based $\beta$ skeletons \cite{kr85} we receive 
a definition based only on a distance criterion:

\begin{definition}
\label{betaskeletons5}
For a given set of points $V$ in $\mathbb{R}^d$ space with a $L_1$ (or $L_{\infty}$) 
metric and for parameters 
$0 \leq \beta \leq \infty$ we define a graph $G_{\beta}(V)$ - called a lens-based 
$\beta$-skeleton - as follows: 
two points $v_1$ and $v_2$ are connected with an edge if and only if  
at least one lens in $N_{1}(v_1,v_2,\beta)$ 
(or $N_{\infty}(v_1,v_2,\beta)$ respectively) does not contain points from 
$V \setminus \{v_1, v_2\}$ where:
\begin{itemize} 
\item
for $0<\beta<1$ a lens $N_{1}(v_1,v_2,\beta)$ (or $N_{\infty}(v_1,v_2,\beta)$ 
respectively) is the intersection of two discs, each of them has radius 
$\frac{d_1(v_1,v_2)}{2 \beta}$ ($\frac{d_{\infty}(v_1,v_2)}{2 \beta}$ respectively), 
whose boundaries contain both $v_1$ and $v_2$ and each shortest
path connecting their centers intersects $S(v_1,v_2)$,
\item
for $1 \leq \beta<\infty$ a lens $N_{1}(v_1,v_2,\beta)$ (or $N_{\infty}(v_1,v_2,\beta)$, 
respectively) is the intersection of two discs with radius 
$\frac{\beta d_{1}(v_,v_2)}{2}$ ($\frac{\beta d_{\infty}(v_,v_2)}{2}$, respectively) 
and with centers $c_1$ and $c_2$, respectively, such that 
$d_1(v_1, c_2)=d_1(v_2,c_1)=|\frac{(\beta -2) d_1(v_1v_2)}{2}|$, 
$d_1(v_1, c_1)=d_1(v_2,c_2)=|\frac{\beta d_1(v_1v_2)}{2}|$ 
($d_{\infty}(v_1, c_2)=d_{\infty}(v_2,c_1)=|\frac{(\beta -2) d_{\infty}(v_1v_2)}{2}|$, 
$d_{\infty}(v_1, c_1)=d_{\infty}(v_2,c_2)=|\frac{\beta d_{\infty}(v_1v_2)}{2}|$ 
in $L_{\infty}$ metric) and each shortest path connecting $c_1$ and $c_2$ intersects
$S(v_1,v_2)$. 

 
\end{itemize}
\end{definition}

Note that the shortest path intersection condition from the definition 
for $0< \beta <1$ and for $1 \leq \beta<\infty$ ensures that centers of the discs 
defining the lens lie antipodically in relation to the set $S(v_1,v_2)$.

The definition \cite{e02} of the circle-based $\beta$-skeletons can also be modified 
in the similar way:\\

\begin{definition}
\label{circle-based} \cite{e02}
For a given set of points $V$ in $\mathbb{R}^d$ space with a $L_1$ (or $L_{\infty}$) 
metric and for parameters $0 \leq \beta < \infty$ we define 
a graph $G_{\beta}(V)$ - called a circle-based $\beta$-skeleton -
as follows: 
two points $v_1$ and $v_2$ are connected with an edge if and only if  
at least one lens in $N_{1}^c(v_1,v_2,\beta)$ (or $N_{\infty}^c(v_1,v_2,\beta)$, 
respectively) does not contain points from $V \setminus \{v_1, v_2\}$ where:
\begin{enumerate}
\item
for $\beta <1$ we define set $N_{1}^c(v_1,v_2,\beta)$ (or $N_{\infty}^c(v_1,v_2,\beta)$ 
respectively) the same way as for lens-based $\beta$-skeleton;
\item
for $1 \leq \beta<\infty$ a lens $N_{1}^c(v_1,v_2,\beta)$ 
(or $N_{\infty}^c(v_1,v_2,\beta)$ respectively) is the sum of two discs, each of them 
has radius $\frac{\beta d_1(v_1,v_2)}{2}$ ($\frac{\beta d_{\infty}(v_1,v_2)}{2}$, 
respectively), whose boundaries contain both $v_1$ and $v_2$ and each shortest path 
connecting their centers intersects $S(v_1,v_2)$. 

\end{enumerate} 
\end{definition}

In our further considerations we will use the following lemma:

\begin{lemma}
\label{cuboid}
Let $C^d_1(p,r)$ ($C^d_{\infty}(p,r)$, respectively) denote a sphere in $\mathbb{R}^d$
space with $L_1$ ($L_{\infty}$, respectively) metric centered in a point $p$ 
and with the radius $r$.
In $d$-dimensional $L_{\infty}$ space, for two points $v_1,v_2$ 
$S(v_1,v_2) = \bigcup_{0\leq r \leq d_i(v_1,v_2)} C^d_i(v_1,r) \cap 
C^d_i(v_2,d_i(v_1,v_2)-r)$, where $i = 1, \infty$.
\end{lemma}

\begin{proof}
It is obvious that for each point $p \in \bigcup_{0\leq r \leq d_i(v_1,v_2)} 
C^d_i(v_1,r) \cap C^d_i(v_2,d_i(v_1,v_2)-r)$, where $i = 1, \infty$ there is 
$d_i(v_1,v_2) = d_i(v_1,p) + d_i(p,v_2)$, i.e. $p \in S(v_1,v_2)$.
If $p \in S(v_1,v_2)$, there exists $r = d_i(v_1,p)$, where $i = 1, \infty$,
such that $d_i(p,v_2) = d_i(v_1,v_2)-r$. Then $p \in C^d_i(v_1,r) \cap 
C^d_i(v_2,d_i(v_1,v_2)-r)$, i.e. $p \in \bigcup_{0\leq r \leq d_i(v_1,v_2)} 
C^d_i(v_1,r) \cap C^d_i(v_2,d_i(v_1,v_2)-r)$.  
\end{proof}


Let $s^d_k$ be a line in $\mathbb{R}^d$ with $L_{\infty}$ metric containing  
$k$-th diagonal of the cube $[0,1]^d$, where $k \in \{1, \cdots, 2^{d-1}$. 
Then $S(v_1,v_2)$ is a polyhedron whose edges are parallel to some line $s^d_k$ 
(see Figure \ref{fig:pajak}). Let $T_k(v_1,v_2)$ be a sum of all lines parallel
to $s^d_k$ intersecting $S(v_1,v_2)$ and $T(v_1,v_2) = \bigcup_k T_k(v_1,v_2)$.
$T(v_1,v_2)$ is a cross with $2^d$ arms (see Figure \ref{fig:pajak}). 

\begin{figure}[htbp]
\centering
\includegraphics[scale=0.4]{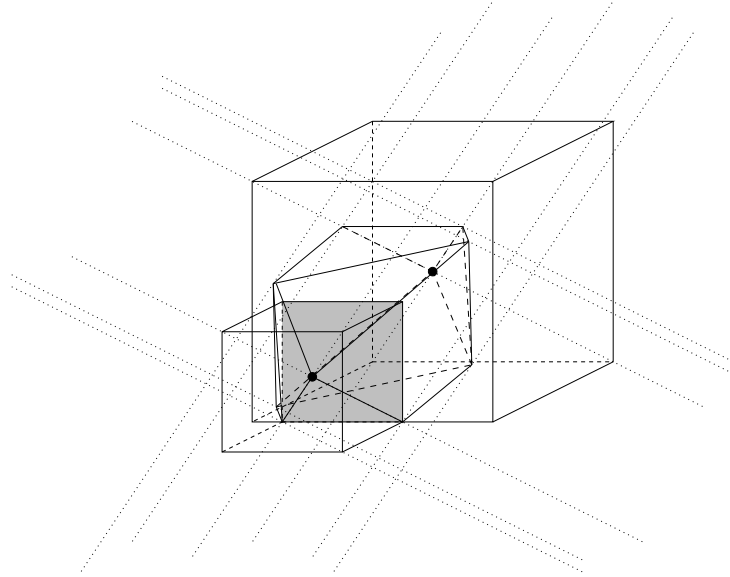}
\caption{Intersection of two spheres in $L_{\infty}$ (grey square), the set $S(v_1,v_2)$
(lines in bold) and two sets $T_{k_1}(v_1,v_2)$ and $T_{k_2}(v_1,v_2)$ (dotted lines).}
\label{fig:pajak}
\end{figure}


In $L_1$ metric the set $S(v_1,v_2)$ is a cuboid (at most $d$-dimensional).
It is a center of the cross $T(v_1,v_2)$ whose each of $d$ arms $T_k(v_1,v_2)$ is
parallel to some coordinate axis (see Figure \ref{fig:krzyz}).

\begin{figure}[htbp]
\centering
\includegraphics[scale=0.4]{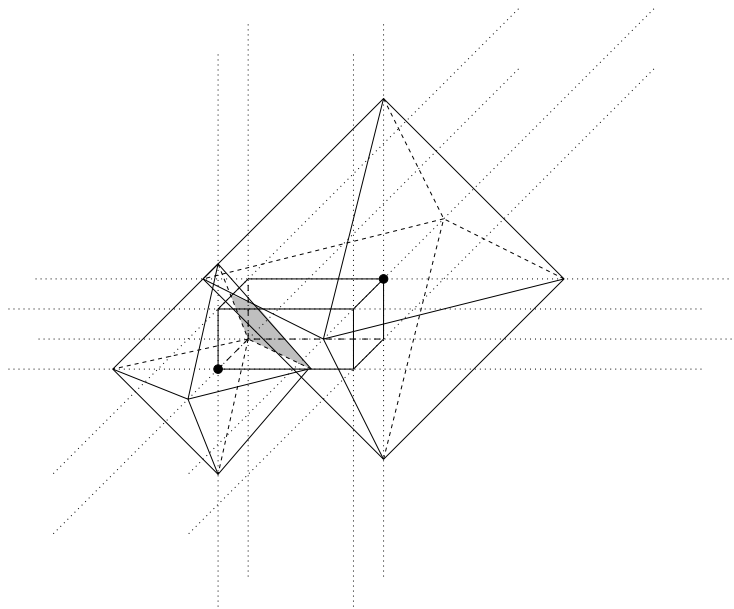}
\caption{Intersection of two spheres in $L_{1}$ (grey area), the set $S(v_1,v_2)$
(lines in bold) and all sets $T_{k_i}(v_1,v_2)$, where $i \in \{1,2,3\}$ (dotted lines).}
\label{fig:krzyz}
\end{figure}

\begin{lemma}\label{intersections}
For $0<\beta<1$ (and for $\beta>1$ for circle-based $\beta$-skeletons) and points 
$v_1,v_2$, let $c_1$, $c_2$ be such points that 
$d_1(c_1,v_1)=d_1(c_2,v_1)=d_1(c_1,v_2)=d_1(c_2,v_2)=\frac{d_{1}(v_1,v_2)}{2 \beta}$ 
($d_{\infty}(c_1,v_1)=d_{\infty}(c_2,v_1)=d_{\infty}(c_1,v_2)=d_{\infty}(c_2,v_2)=
\frac{d_{\infty}(v_1,v_2)}{2 \beta}$, respectively). Then $c_1$ and $c_2$ define 
a lens $N_1(v_1,v_2,\beta)$ ($N_{\infty}(v_1,v_2,\beta)$, respectively) for $v_1,v_2$ 
if and only if both $c_1$ and $c_2$ belong to the set $T_k(v_1,v_2)$ for some
direction $k \in \{1, \cdots, 2^{d-1}\}$ ($k \in \{1, \cdots, d\}$, respectively). 
\end{lemma}
\begin{proof}
Because a distance between $c_1$ and $c_2$ is greater than $d_1(v_1,v_2)$ 
($d_{\infty}(v_1,v_2)$, respectively) then $c_1, c_2 \notin S(v_1,v_2)$.
From the definition of $\beta$-skeleton it follows that $S(c_1,c_2)$ is divided
by $S(v_1,v_2)$ into two separate parts. Hence, there exists an edge direction
of $S(c_1,c_2)$ (e.g. $s^d_k$) that $S(v_1,v_2)$ divides $T_k(c_1,c_2)$.
Hence $T_k(c_1,c_2) \subset T_k(v_1,v_2)$, i.e. $c_1,c_2 \in T_k(v_1,v_2)$.
    
\end{proof}

\section{$\beta$-skeletons in $L_{\infty}$ metric}

We will concentrate on $3$-dimensional space to better describe the solution.
Let $V$ be a set of $n$ points in $3$-dimensional $L_{\infty}$ space.
Let $v_1 = (x_1,y_1,z_1),v_2=(x_2,y_2,z_2)$.


For $\beta=1$ and for points $v_1$ and $v_1$ let $c_1$ and $c_2$ be the centers 
of the discs that define a lens for $v_1,v_2$. 
Depending on how many numbers 
from $|x_1-x_2|,|y_1-y_2|,|z_1-z_2|$ is equal to the distance $d_{\infty}(v_1,v_2)$ 
the intersection of spheres $C(v_1,\frac{d_{\infty}(v_1,v_2)}{2})$ and 
$C(v_2,\frac{d_{\infty}(v_1,v_2)}{2})$ can be a rectangle, a segment or a point.


For each point 
$v \in C(v_1,\frac{d_{\infty}(v_1,v_2)}{2}) \cap C(v_2,\frac{d_{\infty}(v_1,v_2)}{2})$ 
there is 
$d_{\infty}(v,v_1)+d_{\infty}(v,v_2)=d_{\infty}(v_1,v_2)$, 
so this all such points belong to $S(v_1,v_2)$.


Let $c_1,c_2$ be two points from the intersection set such that the distance 
between them is maximal. In each case there is at most two pairs of points like that
(in rectangle both pairs of such points define the same lens). 
Then the lens defined by $c_1$ and $c_2$ is minimal.

For a given $\beta \in(1,2)$, let $C_1$ ($C_2$ respectively) be the intersection 
of spheres $C(v_1,\frac{\beta d_{\infty}(v_1,v_2)}{2})$ and 
$C(v_2,(1 -\frac{\beta}{2})d_{\infty}(v_1,v_2))$ 
($C(v_2,\frac{\beta d_{\infty}(v_1,v_2)}{2})$ and 
$C(v_1,(1 -\frac{\beta}{2})d_{\infty}(v_1,v_2))$ respectively).  
Note that just like for $\beta=1$ set $C_1$ ($C_2$ respectively) can be a rectangle, 
a segment or a point depending on how many numbers from $|x_1-x_2|,|y_1-y_2|,|z_1-z_2|$ 
is equal to the distance $d_{\infty}(v_1,v_2)$. Sets $C_1$ and $C_2$ are symmetrical 
and parallel to each other and also parallel to one or more of the $\{x=0\},\{y=0\}$ 
or $\{z=0\}$ planes.

Note that since for any point $v$ in the set 
$C(v_1,\frac{\beta d_{\infty}(v_1,v_2)}{2}) \cap 
C(v_2,(1 -\frac{\beta}{2})d_{\infty}(v_1,v_2))$ 
($C(v_2,\frac{\beta d_{\infty}(v_1,v_2)}{2}) \cap 
C(v_1,(1 -\frac{\beta}{2})d_{\infty}(v_1,v_2))$ respectively) we have 
$d_{\infty}(v,v_1)+d_{\infty}(v,v_2)=
\frac{\beta d_{\infty}(v_1,v_2)}{2}+(1 -\frac{\beta}{2})d_{\infty}(v_1,v_2)=
d_{\infty}(v_1,v_2)$ and $\frac{\beta d_{\infty}(v_1,v_2)}{2} \leq d_{\infty}(v_1,v_2)$. 
Hence, each point from the set $C_1$ ($C_2$ respectively) also belongs to the set 
$S(v_1,v_2)$. 
A pair $c_1 \in C_1, c_2 \in C_2$ defines the minimal lens 
if the distance between them is maximal. In some cases there can be more than one 
such pair, but all those pairs define the same lens.




If $\beta=2$, the lens for points $v_1,v_2$ is defined uniquely. 

For $0<\beta<1$ in case of the lens-based $\beta$-skeleton (or for $\beta>1$ for 
circle-based ones), all possible centers of the discs defining the lenses belong 
to the intersection of two spheres $C(v_1,\frac{d_{\infty}(v_1,v_2)}{2 \beta})$ 
and $C(v_2,\frac{d_{\infty}(v_1,v_2)}{2 \beta})$. If segment $v_1v_2$ is not parallel 
to a coordinate axis, then this intersection is a curve composed of six edges 
of the cuboid    
$D(v_1,\frac{d_{\infty}(v_1,v_2)}{2 \beta}) \cap D(v_2,\frac{d_{\infty}(v_1,v_2)}
{2 \beta})$, where $D(p,r) = \{q \in \mathbb{R}^3: d_{\infty}(p,q) \leq r\}$ . 
In other case, the intersection is not a curve, but a band.

Note that for $\beta=1$ the intersection of spheres that defines centers 
of lenses is a rectangle that is equal to a cross-section of the set $S(v_1,v_2)$. 
For $1>\beta \rightarrow 0$ (or $1 < \beta \rightarrow \infty$ for 
circle-based skeltons)vertices of this rectangle travel away from 
the set $S(v_1,v_2)$ along the lines parallel to diagonals of a $3$-dimensional cube. 
For $\beta=1$ vertices of the curve 
$C(v_1,\frac{d_{\infty}(v_1,v_2)}{2 \beta}) \cap C(v_2,\frac{d_{\infty}(v_1,v_2)}
{2 \beta})$ are in the set $S(v_1,v_2)$ and that for $\beta \rightarrow 0$ 
($1 < \beta \rightarrow \infty$, respectively) they move 
further away from $S(v_1,v_2)$ inside arms of the set $T(v_1,v_2)$. Therefore, 
the intersection of the curve 
$C(v_1,\frac{d_{\infty}(v_1,v_2)}{2 \beta}) \cap C(v_2,\frac{d_{\infty}(v_1,v_2)}
{2 \beta})$ and set $T(v_1,v_2)$ are respective vertices of the cuboid 
$D(v_1,\frac{d_{\infty}(v_1,v_2)}{2 \beta}) \cap D(v_2,\frac{d_{\infty}(v_1,v_2)}
{2 \beta})$ along with short segments, parts of the intersection curve, incident 
to those vertices. If $c_1,c_2$ are endpoints of a diagonal of the cuboid 
$D(v_1,\frac{d_{\infty}(v_1,v_2)}{2 \beta}) \cap D(v_2,\frac{d_{\infty}(v_1,v_2)}
{2 \beta})$ and they are both in the curve 
$C(v_1,\frac{d_{\infty}(v_1,v_2)}{2 \beta}) \cap C(v_2,\frac{d_{\infty}(v_1,v_2)}
{2 \beta})$ then the lens they define is minimal.

For $\beta>1$ in the case of circle-based skeleton we compute the pairs of closest
centers $c_1$ and $c_2$. It can be done in constant time. We also test emptiness 
of both spheres centered in $c_1$ and $c_2$ instead of only their intersection.


If segment $v_1v_2$ is parallel to an axis then the intersection of the band that 
contains the potential centers of the spheres that define lenses and the arms 
of the set $T(v_1,v_2)$ is a sum of eight segments, four on each side of the band.


For $\beta>2$ let us 
consider the set $C_1$ ($C_2$ respectively) which is the intersection of spheres 
$C(v_1,\frac{\beta d_{\infty}(v_1,v_2)}{2})$ and 
$C(v_2,(1 -\frac{\beta}{2})d_{\infty}(v_1,v_2))$ 
($C(v_2,\frac{\beta d_{\infty}(v_1,v_2)}{2})$ and 
$C(v_1,(\frac{\beta}{2}-1)d_{\infty}(v_1,v_2))$ respectively). Note that since 
$(\frac{\beta}{2}-1)d_{\infty}(v_1,v_2)+d_{\infty}(v_1,v_2)=
\frac{\beta}{2}d_{\infty}(v_1,v_2)$ then spheres 
$C(v_1,\frac{\beta d_{\infty}(v_1,v_2)}{2})$ and 
$C(v_2,(1 -\frac{\beta}{2})d_{\infty}(v_1,v_2))$ 
($C(v_2,\frac{\beta d_{\infty}(v_1,v_2)}{2})$ and 
$C(v_1,(\frac{\beta}{2}-1)d_{\infty}(v_1,v_2))$ respectively) are tangent internally 
and the sphere with the smaller radius is contained in the bigger one. From this 
we get that the set $C_1$ ($C_2$ respectively) is a sum of up to three respective 
faces of the sphere with the smaller radius, depending on the number of coordinates 
that define the distance $d_{\infty}(v_1,v_2)$. Intersection of $C_1$ 
($C_2$ respectively) and $T(v_1,v_2)$ is a subset of vertices of $C_1$ 
($C_2$ respectively).


Based on the observations in this section the following theorem can be proved:

\begin{theorem}
Let $V \in \mathbb{R}^3$ in $L_{\infty}$ metric be a set of $n$ points in general 
position. Then, for $\beta<2$ $\beta$-skeleton $G_{\beta}(V)$ can 
be computed in $O(n^2 \log^3 n)$ time. For $\beta \geq 2$ there exists 
an $O(n \log^2 n)$ time algorithm.
\end{theorem}
\begin{proof}
Since all lenses in $L_{\infty}$ metric are cuboids, range search algorithms 
\cite{s89,bk97} can be used to compute $\beta$-skeletons. Note that for $\beta<2$ it 
is necessary to check all $O(n^2)$ edges, but for each edge the minimal lens is 
uniquely defined.
For $\beta \geq 2$, since $G_{\beta}(V)$ is a subgraph of the euclidean relative 
neighborhood graph, then $G_{\beta}(V)$ has $O(n)$ edges. For each edge there is 
a constant $O(C)$ number of minimal lenses, so the lens-based $\beta$-skeleton can 
be computed in $O(n \log^2 n)$ time \cite{s89}.
\end{proof}

Used the same methods we obtain the following result.

\begin{theorem}
Let $V \in \mathbb{R}^d$ in $L_{\infty}$ metric be a set of $n$ points in general 
position. Then, for $\beta<2$ $\beta$-skeleton $G_{\beta}(V)$ can 
be computed in $O(n^2 \log^d n)$ time. For $\beta \geq 2$ there exists 
an $O(n \log^{d-1} n)$ time algorithm.
\end{theorem}

For $\beta' > \beta$ there is $G_{\beta'}(V) \subset G_{\beta}(V)$.
Moreover, for a pair of points $v_1,v_2$ and $i < \infty$ the following
inclusion holds $N_i(v_1,v_2,\beta) \subset N_{\infty}(v_1,v_2,\beta)$. 
Due result of Agarwal and Matousek \cite{am92} we obtain the following theorem.

\begin{theorem}
Let $V \in \mathbb{R}^3$ in $L_{\infty}$ metric be a set of $n$ points in arbitrary 
position. Then, for $\beta \geq 2$ and each $\epsilon > 0$ there exists 
an $O(n^{\frac{7}{4+\epsilon}} \log^{d} n)$ time algorithm that constructs 
$\beta$-skeleton for the set $V$.
\end{theorem} 

\section{$\beta$-skeletons in $L_{1}$ metric}
Now, let us consider the problem for a set $V$ of $n$ points in $\mathbb{R}^d$ 
with the $L_1$ metric. We will concentrate on $3$-dimensional space to better describe 
our considerations. Let $v_1=(x_1,y_1z_1), v_2=(x_1,y_2,z_2)$.

For $\beta<1$ (and for $\beta>1$ in case of circle-based $\beta$-skeletons) 
and for points $v_1=(x_1,y_1z_1), v_2=(x_1,y_2,z_2)$ the shape of the intersection 
$C(v_1,\frac{d_{1}(v_1,v_2)}{2 \beta})$ and $C(v_2,\frac{d_{1}(v_1,v_2)}{2 \beta})$ 
depends on whether $|x_1-x_2|=|y_1-y_2|$. If $|x_1-x_2|\neq |y_1-y_2|$ this 
intersection is a parallelogram 
which common part with the cross $T(v_1,v_2)$ consists of two pairs of parallel segments. 
For each pair two most distant points define the minimal lens.  

If $|x_1-x_2|=|y_1-y_2| \neq 0$ then set 
$C(v_1,\frac{d_{1}(v_1,v_2)}{2 \beta}) \cap C(v_2,\frac{d_{1}(v_1,v_2)}{2 \beta})$ 
consists of two parallel segments. Cross $T(v_1,v_2)$ intersects each of those segments 
and a pair of most distant centers, one from one segment and one from another, 
defines the minimal lens.  

In the last case, when $|x_1-x_2|=|y_1-y_2|=0$ the intersection is a square. 
Set $T(v_1,v_2)$ intersects it in four points but the lenses for two pairs of centers 
defined this way are identical. 

For $\beta \in [1,2)$ there are also three cases to consider. If segment 
$v_1v_2$ is parallel to some coordinate axis then the intersection of spheres 
$C(v_1,\frac{\beta d_{1}(v_1,v_2)}{2})$ and $C(v_2,\frac{(\beta-2) d_{1}(v_1,v_2)}{2})$ 
(spheres $C(v_2,\frac{\beta d_{1}(v_1,v_2)}{2})$ and 
$C(v_1,\frac{(\beta-2) d_{1}(v_1,v_2)}{2})$ respectively) is just one point. 
From the definition, it lies in the set $S(v_1,v_2)$ and the lens for points $v_1,v_2$ 
is uniquely defined.

If $|z_1-z_2|=0$ but segment $v_1v_2$ is not parallel to any coordinate axis then 
the intersection is a single segment in the plane $\{z=z_1\}$. This segment is fully 
contained in the cuboid $S(v_1,v_2)$. Since the intersection of spheres 
$C(v_2,\frac{\beta d_{1}(v_1,v_2)}{2})$ and $C(v_1,\frac{(\beta-2) d_{1}(v_1,v_2)}{2})$ 
is a segment parallel to this one, that also lies in the cuboid $S(v_1,v_2)$, 
the pair of points $c_1$ and $c_2$, one from one segment, one from the other, 
define the minimal lens if the euclidean distance between them is maximal. 
Since the segments described above are symmetrical and parallel it is only necessary 
to consider their endpoints.

If $|z_1-z_2| \neq 0$ and $v_1v_2$ is not parallel to a coordinate axis then 
the intersection 
$C(v_1,\frac{\beta d_{1}(v_1,v_2)}{2}) \cap C(v_2,\frac{(\beta-2) d_{1}(v_1,v_2)}{2})$ 
is a hexagon with edges pairwise parallel.  Note that since for any point $v$ in this 
hexagon it is true that 
$d_1(v,v_1)+d_1(v,v_2)=\frac{\beta d_{1}(v_1,v_2)}{2}+\frac{(\beta-2) d_{1}(v_1,v_2)}{2}
=d_1(v_1,v_2)$, so this polygon is contained in the set $S(v_1,v_2)$. The intersection 
of spheres $C(v_2,\frac{\beta d_{1}(v_1,v_2)}{2})$ and 
$C(v_1,\frac{(\beta-2) d_{1}(v_1,v_2)}{2})$ is also a symmetrical hexagon, parallel 
to the previous one, that is also in $S(v_1,v_2)$. A pair of points $c_1,c_2$ one 
from one hexagon, second from the other such that the euclidean distance between them 
is maximal define the minimal lens. There can be up to six pairs like that.

For $\beta=2$ the centers of the discs defining the lenses are uniquely defined.

For $\beta>2$ let us denote by $C_1$ ($C_2$ respectively) the intersection 
of spheres $C(v_1,\frac{\beta d_{1}(v_1,v_2)}{2})$ and 
$C(v_2,\frac{(\beta-2) d_{1}(v_1,v_2)}{2})$ ($C(v_2,\frac{\beta d_{1}(v_1,v_2)}{2})$ 
and $C(v_1,\frac{(\beta-2) d_{1}(v_1,v_2)}{2})$ respectively). $C_1$ can be equal 
to either one face of the sphere with the smaller radius or to a sum of two or even 
four such faces. 
If points $v_1$ and $v_2$ have all coordinates different 
($|x_1-x_2|,|y_1-y_2|,|z_1-z_2| \neq 0$) then this intersection is one of the faces 
from sphere with radius $\frac{\beta d_{1}(v_1,v_2)}{2}$. Set $T(v_1,v_2)$ intersects 
this faces in three points (in its vertices). Each such point belongs to a different 
arm of $T(v_1,v_2)$. In set $C_2$ there are also only three points, so for each pair 
of opposing arms exists a pair $c_1,c_2$ that defines a lens.

If $v_1$ and $v_2$ have one of the coordinates equal then $C_1$ is a sum of two 
respective faces. The cross intersects it in vertices of those two triangles. 
Therefore, there are $4$ pairs of lenses defined for $v_1,v_2$.

In the last case, when segment $v_1v_2$ is parallel to one of the coordinates axis, 
the intersection 
$C(v_1,\frac{\beta d_{1}(v_1,v_2)}{2}) \cap C(v_2,\frac{(\beta-2) d_{1}(v_1,v_2)}{2})$ 
($C(v_2,\frac{\beta d_{1}(v_1,v_2)}{2}) \cap C(v_1,\frac{(\beta-2) d_{1}(v_1,v_2)}{2})$ 
respectively) is a sum of four neighbouring faces. Only vertices of those faces belong 
to arms of the set $T(v_1,v_2)$. There are five points in set $C_1$ ($C_2$ respectively)
and therefore there are $5$ lenses defined for $v_1,v_2$. Note that in each case 
the pair $c_1 \in C_1,c_2 \in C_2$ such that the distance between them is maximal 
defines the smallest lens.

The algorithm constructing $\beta$-skeleton for a given parameter $\beta$
is similar to the previous one. We use range trees to test emptiness of lenses. 

However, in this case analyzed regions are not cuboids. Therefore, we modify 
range tree structure by adding two new coordinates corresponding to each pair directions
of sphere faces in $L_1$. In this way we obtain a finite number of data
structures. We test intersections of spheres analyzing pairs of their parts each 
of which is limited by planes parallel to the coordinate axes and one plane parallel
to face of the sphere.

Based on the observations in this section the following theorem can be proved:

\begin{theorem}
Let $V \in \mathbb{R}^d$ in $L_1$ metric be a set of $n$ points in arbitrary 
position. Then, $\beta$-skeleton $G_{\beta}(V)$ can 
be computed in $O(n^2 \log^{d+2} n)$ time. 
\end{theorem}


\section{Conclusions and open problems}

We have described shapes of lenses for $\beta$-skeletons in $\mathbb{R}^d$ space 
with $L_1$ and $L_{\infty}$ metric. The definition we used is based on the distance 
criterion and on the idea of choosing pairs of antipodal points by computing the sets 
of shortest paths between them. Since for each pair 
$v_1,v_2 \in V \subseteq \mathbb{R}^d$ and for every $\beta$ there is no more then 
a constant number of minimal lenses (and this constant number does not depend 
on the number of points in the set $V$) it is possible to use range search algorithm 
to compute the $\beta$-skeleton. 
Since in euclidean metric in $\mathbb{R}^3$ a size of Gabriel Graph 
is $O(n^{\frac{3}{2}})$ \cite{kl10}, we think that the algorithm for 
$1 \leq \beta \leq 2$ in $L_{\infty}$ metric presented in this work could be 
improved. Moreover, algorithms in $L_1$ seem to be not optimal.


\bibliographystyle{abbrv}

\end{document}